\documentclass[a4paper,10pt]{article}

\usepackage{microtype}
\usepackage[T1]{fontenc}
\usepackage[utf8]{inputenc}
\usepackage{hyperref}
\usepackage{mathtools}
\usepackage{amsmath}
\usepackage{amssymb}

\usepackage{fourier}
\usepackage[semibold,scale=0.92]{FiraSans}
\usepackage[scale=0.92,lf]{FiraMono}
\usepackage[small,sf,bf]{titlesec}

\usepackage{amsthm}
\usepackage{thm-restate}

\newtheorem{theorem}{Theorem}
\newtheorem{lemma}[theorem]{Lemma}
\newtheorem{corollary}[theorem]{Corollary}
\theoremstyle{definition}
\newtheorem{definition}[theorem]{Definition}

\newcommand{\Os}{O^{*}\!}
\newcommand{\N}{\mathbb{N}}
\newcommand{\Qnn}{\mathbb{Q}_{\geqslant 0}}

\newcommand{\pf}{\operatorname{pf}}
\newcommand{\coef}{\operatorname{coef}}
\newcommand{\sgn}{\operatorname{sgn}}

\usepackage{nicefrac}

\usepackage{tabularx} % better tables
\usepackage{todonotes}

\renewcommand\footnoterule{\kern-3pt \hrule width 2in height 0.25pt \kern 2.75pt}

\newenvironment{problem}[1]{
  \vskip 1em \small \begin{tabular}{@{}ll} \multicolumn{2}{@{}l}{\textsf{\textbf{#1}}}\\[3pt]
}{
  \end{tabular} \vskip 1em
}

\title{Tight Vector Bin Packing with Few Small Items via Fast Exact Matching in Multigraphs}

\author{Alexandra Lassota \\ {\normalsize EPFL} \and Aleksander Łukasiewicz \\ {\normalsize University of Wrocław} \and Adam Polak \\ {\normalsize EPFL}}

\date{}

\begin{document}

\maketitle

\begin{abstract}
We solve the Bin Packing problem in $O^*(2^k)$ time, where $k$ is the number of items less or equal to one third of the bin capacity. This parameter measures the distance from the polynomially solvable case of only large (i.e., greater than one third) items. Our algorithm is actually designed to work for a more general Vector Bin Packing problem, in which items are multidimensional vectors. We improve over the previous fastest $O^*(k! \cdot 4^k)$ time algorithm.

Our algorithm works by reducing the problem to finding an exact weight perfect matching in a (multi-)graph with $O^*(2^k)$ edges, whose weights are integers of the order of $O^*(2^k)$. To solve the matching problem in the desired time, we give a variant of the classic Mulmuley-Vazirani-Vazirani algorithm with only a~linear dependence on the edge weights and the number of edges -- which may be of independent interest.

Moreover, we give a tight lower bound, under the Strong Exponential Time Hypothesis (SETH), showing that the constant $2$ in the base of the exponent cannot be further improved for Vector Bin Packing.

Our techniques also lead to improved algorithms for Vector Multiple Knapsack, Vector Bin Covering, and Perfect Matching with Hitting Constraints.
\end{abstract}

\section{Introduction}
NP-hard problems often have special cases that can be solved in polynomial time, e.g., Vertex Cover is tractable in graphs with the Kőnig property, Dominating Set is tractable in trees, and Longest Common Subsequence is tractable in permutations.
Many of these problems remain (fixed-parameter) tractable with a distance from the polynomially solvable case taken as a parameter, e.g., Vertex Cover Above Matching~\cite{RazgonO08}, Dominating Set in bounded treewidth graphs~\cite{AlberBFKN02}, or Longest Common Subsequence parameterized by the maximum occurrence number~\cite{GuoHN04}.
In parameterized complexity, this concept is sometimes dubbed \emph{distance from triviality}~\cite{GuoHN04}.

In the Bin Packing problem, we are given $n$ \emph{items} from $\Qnn$, and we have to pack them into the smallest possible number of unit-sized \emph{bins}. It is a classic strongly NP-hard problem. When all items are \emph{large}, i.e., greater than $\nicefrac{1}{3}$, then no three items can fit into a single bin and the problem reduces to the Maximum Matching problem, and hence it can be solved in polynomial time~\cite{Edmonds65}.

Bannach et al.~\cite{Bannach20} are the first to study Bin Packing parameterized by the number $k$ of \emph{small} (i.e., $\leqslant \nicefrac{1}{3}$) items -- which is the distance from the above tractable case.
They give algorithms running in randomized $\Os(k! \cdot 4^k)$ time, and deterministic $\Os\big((k!)^2 \cdot 2^k\big)$ time.\footnote{We use $\Os(\cdot)$ notation to suppress factors polynomial in the input size $n$, i.e., $\Os(f(k)) = f(k) \cdot n^{O(1)}$\!.}
Their randomized algorithm works even for a more general Vector Bin Packing problem, in which items are $d$-dimensional vectors from $\Qnn^d$, and a set of items fits into a bin if their coordinate-wise sum does not exceed $1$ in any coordinate.
(The notion of a small item is more complex in the multidimensional case; see Section~\ref{sec:definition} for the definition.)

We improve upon their result by giving an $\Os(2^k)$ randomized time algorithm for Vector Bin Packing.
We complement it with a matching conditional lower bound, showing that the constant $2$ in the base of the exponent cannot be further improved, unless the Strong Exponential Time Hypothesis ({\small SETH}) fails.

Our algorithm works by reducing the problem to finding a perfect matching of a~given total weight in an edge-weighted (multi-)graph. The graph has only $O(n)$ nodes, but can have up to $O(2^kn^2)$ edges, whose weights are integers of the order of $2^k \cdot k$. To solve the matching problem in the desired $\Os(2^k)$ time, we give a variant of the classic Mulmuley-Vazirani-Vazirani algorithm~\cite{Mulmuley1987-mc} with only a linear dependence on the edge weights and the number of edges -- which may be of independent interest.

Our techniques also lead to improved algorithms for the two other problems studied by Bannach et al.~\cite{Bannach20}, i.e., the Vector Multiple Knapsack and Vector Bin Covering problems, as well as for the Perfect Matching with Hitting Constraints problem, studied by Marx and Pilipczuk~\cite{Marx14}.

\subsection{Vector Bin Packing with Few Small Items}
\label{sec:definition}

First, let us formally define Vector Bin Packing as a decision problem. We remark that (Vector) Bin Packing is also often studied as an optimization problem -- especially in the context of approximation algorithms -- but one can always switch between the two variants via binary search, loosing at most a factor of $O(\log n)$ in the running time.

\begin{problem}{Vector Bin Packing}
  Given:
    & a set of $n$ items $V = \{v_1, \ldots, v_n\} \subseteq \Qnn^d$, \\
    & and an integer $\ell \in \mathbb{Z}_+$ denoting the number of unit-sized bins. \\[3pt]
  Decide:
    & if the items can be partitioned into $\ell$ bins $B_1 \cup \cdots \cup B_\ell = V$ such that \\
    & $\sum_{v \in B_i} v[j] \leqslant 1$ for every bin $i \in [\ell]$ and every dimension $j \in [d]$.\footnotemark
\end{problem}
\footnotetext{We use $[n]$ to denote the set of integers $\{1,2,\ldots,n\}$.}

Note that the assumption that bins are unit-sized is without loss of generality, as one can always independently scale each dimension in order to meet that constraint. We can also safely assume that $V$ is a set, as one can handle multiple occurrences of the same item by introducing one extra dimensions with negligibly small but unique coordinates.

Unlike in the one-dimensional Bin Packing problem, where a small item can be defined simply as smaller or equal to $\nicefrac{1}{3}$, we use a more complex definition, introduced by Bannach et al.~\cite{Bannach20}.
Let $V \subseteq \Qnn^d$ be a set of $d$-dimensional items. We say that a subset $V' \subseteq V$ is \textbf{3-incompatible} if no three distinct items from $V'$ fit into a unit-sized bin, i.e., for every distinct $u, v, w \in V'$ there exists a dimension $i \in [d]$ such that $u[i]+v[i]+w[i] > 1$.
Now we can define the parameterized problem that we study.

\begin{problem}{Vector Bin Packing with Few Small Items}
  Parameter:
    & the number of \emph{small} items $k$. \\[3pt]
  Given:
    & a set of $n$ items $V = \{v_1, \ldots, v_n\} \subseteq \Qnn^d$, \\
    & a subset of $k$ items $V_S \subseteq V$ such that $V_L = V \setminus V_S$ is \textbf{3-incompatible}, \\
    & and an integer $\ell \in \mathbb{Z}_+$ denoting the number of unit-sized bins. \\[3pt]
  Decide:
    & if the items can be partitioned into $\ell$ bins $B_1 \cup \cdots \cup B_\ell = V$ such that \\
    & $\sum_{v \in B_i} v[j] \leqslant 1$ for every bin $i \in [\ell]$ and every dimension $j \in [d]$.
\end{problem}

We say that items in $V_S$ are \emph{small}, and the remaining items in $V_L = V \setminus V_S$ are \emph{large}. Note that we assume that a subset of small items is specified in the input. This way we can study the complexity of the packing problem independently of the complexity of finding a (smallest) subset of small items. This is similar, e.g., to the standard practice for treewidth parameterization, where one assumes that a suitable tree decomposition is given in the input (see, e.g., \cite{CyganFKLMPPS15}). We remark that if only the set of all items $V$ is given, a smallest possible subset of small items can be found in $\Os(2.0755^k)$ time~\cite{Wahlstrom07} via a reduction to the 3-Hitting Set problem~\cite{Bannach20}.

\subsection{Our results}

Our main result is an $\Os(2^k)$ time randomized algorithm for Vector Bin Packing with Few Small Items. The algorithm consists of two parts: reducing the packing problem to a matching problem, and solving the matching problem. More formally, we first prove the following.

\begin{restatable}{lemma}{lempacktomatch}
\label{lem:packtomatch}
An $n$-item instance of Vector Bin Packing with $k$ small items can be reduced, in deterministic time $O(2^kn^2kd)$, to the problem of finding an exact-weight perfect matching in a (multi-)graph. The graph has $O(n)$ vertices, $O(2^kn^2)$ edges, and non-negative integer edge weights that do not exceed $O(2^kk)$. The target exact total weight of a matching is $O(2^kk)$.
\end{restatable}

The above matching problem is dubbed Exact Matching, and is known to be in randomized\footnote{It is a big open problem to derandomize the algorithm, see, e.g., \cite{SvenssonT17}.} (pseudo-)polynomial time since the Mulmuley-Vazirani-Vazirani algorithm~\cite{Mulmuley1987-mc}. The algorithm directly solves the $0$/$1$ weights variant of Exact Matching in simple graphs. A prior reduction of Papadimitriou and Yannakakis~\cite{Papadimitriou1982-qg} handles arbitrary non-negative integer edge weights and multiple parallel edges. The reduction replaces each edge of weight $w$ by a path of length $2w+1$ with alternating $0$/$1$ edge weights.

The reduction multiplies the number of vertices by the edge weights and by the number of edges. Further, Mulmuley-Vazirani-Vazirani is not a linear time algorithm. Hence, this would give us only a $2^{O(k)}n^{O(1)}$ time algorithm for Vector Bin Packing. This is already an improvement over the previous factorial time algorithm, but still not our desired $2^k n^{O(1)}$ running time.

There are more direct and faster ways to solve the general Exact Matching problem than going through the Papadimitriou-Yannakakis reduction.
It seems folklore to handle arbitrary edge weights by replacing a monomial $x$, corresponding to a weight-one edge in the Mulmuley-Vazirani-Vazirani algorithm, with $x^w$, where $w$ is the edge weight.
It remains to handle multiple parallel edges.
A crucial part of the algorithm is the so-called \emph{isolation lemma}.
It assigns random \emph{costs} to edges, ensuring that the minimum cost perfect matching of the target weight is unique, and hence it cannot cancel out in the algebraic computations.
The range of costs, required to ensures that property, on one hand depends on the number of edges, and on the other hand, determines the bitsize of the costs, on which the algorithm later needs to do arithmetic.

Due to the number of edges in Lemma~\ref{lem:packtomatch}, a direct application of isolation lemma would lead to an $\Os(4^k)$ time algorithm. To mitigate this issue, we carefully apply isolation lemma to pairs of vertices, and hence the number of edges appears in the running time only as a linear additive factor.

\begin{restatable}{theorem}{thmmatching}
\label{thm:matching}
  Given an edge-weighted multigraph with $n$ nodes and $m$ edges, and an integer~$t$,
  a randomized Monte-Carlo algorithm can decide whether there is
  a perfect matching of total weight exactly $t$ in $\widetilde O(t \cdot n^8 + m)$ time.
\end{restatable}

Lemma~\ref{lem:packtomatch} and Theorem~\ref{thm:matching} together imply our main result.

\begin{theorem}
\label{thm:binpacking}
There is a randomized Monte Carlo algorithm solving Vector Bin Packing with Few Small Items in $\Os(2^k)$ time.
\end{theorem}

In Appendix~\ref{sec:alt} we give an alternative proof of Theorem~\ref{thm:binpacking}, using a different algorithm, whose running time has a better dependence on the number of items $n$. This algorithm, however, presents a more complicated and tailored approach; in particular, it does not seem to generalize to the Vector Bin Covering problem that we discuss later.

\subsubsection*{Lower bound}

We show that the above result is tight, via a matching conditional lower bound, under the Strong Exponential Time Hypothesis ({\small SETH})~\cite{ImpagliazzoP01}. The hypothesis states that deciding $k${\small-CNF-SAT} with $n$ variables requires time $2^{s_k n}$ for $\lim_{k \to \infty} s_k = 1$. In particular, it implies that deciding {\small CNF-SAT} requires $2^{(1-\varepsilon) n}$ time, for every $\varepsilon > 0$. {\small SETH} is a standard hardness assumption for conditional lower bounds in fine-grained and parameterized complexity~\cite{CyganFKLMPPS15,williams2018some}. We prove the following lower bound for the (non-parameterized) Vector Bin Packing problem.

\begin{restatable}{theorem}{thmlowerbound}
Unless {\small SETH} fails, Vector Bin Packing cannot be solved in $\Os(2^{(1-\varepsilon) n})$ time, for any $\varepsilon > 0$. This holds even restricted to instances with only two bins and dimension $d = O(n)$.
\end{restatable}

Since $k \leqslant n$, the corollary for the parameterized version of the problem follows immediately, proving that the algorithm of Theorem~\ref{thm:binpacking} is tight.

\begin{corollary}
Unless {\small SETH} fails, Vector Bin Packing with Few Small Items cannot be solved in $\Os(2^{(1-\varepsilon) k})$ time, for any $\varepsilon > 0$. This holds even restricted to instances with only two bins and dimension $d = O(n)$.
\end{corollary}

We remark that our lower bound crucially relies on multiple dimensions. The best known hardness result for the (one-dimensional) Bin Packing problem rules out only $2^{o(n)}$~time algorithms~\cite{JansenLL16}, assuming the Exponential Time Hypothesis ({\small ETH})~\cite{Impagliazzo98}. It is a big open problem whether an $O(1.99^n)$ time algorithm for Bin Packing exists. Recently, Nederlof et al.~\cite{NederlofPSW21} gave such an algorithm for any constant number of bins. This is in contrast to Vector Bin Packing, which, as we show, requires $2^{(1-\varepsilon) n}$ time already for two bins.

\subsubsection*{Other applications}
Bannach et al.~\cite{Bannach20} studied two further problems closely related to the Vector Bin Packing problem -- namely, Vector Multiple Knapsack and Vector Bin Covering -- under similar parameterizations.

In the Vector Multiple Knapsack problem, each item comes with a \emph{profit}, and instead of having to pack all the items, we aim to pack a subset of the items into a fixed number of bins while maximizing the overall profit of the packed items. In the few small items regime, the fastest known algorithm so far has a running time of $\Os(k! \cdot 4^k)$, where $k$ is the number of small items~\cite{Bannach20}. Adapting our algorithm to handle the profits and the obstacle that only a subset of items might be packed, we obtain the following theorem.

\begin{restatable}{theorem}{thmKnapsack}
  There is a randomized Monte Carlo algorithm solving Vector Multiple Knapsack with Few Small Items in $\Os(2^k)$ time when item profits are bounded by $\operatorname{poly}(n)$.
\end{restatable}

In the Vector Bin Covering problem, we aim to \emph{cover} bins. Intuitively speaking, instead of packing the items into as few bins as possible, we want to partition them into as many bins as possible while satisfying a \emph{covering constraint} for each bin. This new desired property of a solution leads to a slightly different definition of the set of small items: instead of any three large items not fitting together into a bin, now they cover a bin. So far, the fastest algorithm solving this problem parameterized by the number $k$ of small items\footnote{Even though Bannach et al.~do not explicitly adapt their definition of a small item to this problem, they indeed work with the same definition as we do. In the full version on arXiv~\cite[page 11]{Bannach20Long} they write: ``The large vectors have the property that every subset of three vectors cover a container.''} runs in time $\Os(k! \cdot 4^k)$~\cite{Bannach20}. We give the following improvement.

\begin{restatable}{theorem}{thmCover}
  There is a randomized Monte Carlo algorithm solving Vector Bin Covering with Few Small Items in $\Os(2^k)$ time.
\end{restatable}

Further, our results directly imply an improved running time for the Perfect Matching with Hitting Constraints problem. This problem asks whether we can find a~perfect matching in a graph using at least one edge from each of given subsets of edges.
It was studied by Marx and Pilipczuk~\cite{Marx14} as a tool for solving a subgraph isomorphism problem in forests.
They gave an algorithm (for the matching problem) running in time $2^{O(k)}n^{O(1)}$, where $k$ is the number of edge subsets.
Their algorithm shares certain similarities with our Vector Bin Packing algorithm.
They use, however, a~less efficient encoding of subsets into edge weights (using $2k$ bits, compared to $k \log k$ bits we achieve in Lemma~\ref{lem:edgeWeights}), and they only coarsely analyze the polynomial dependence on the weights when solving Exact Matching. 
Avoiding these two inefficiencies, we prove the following theorem.

\begin{restatable}{theorem}{thmPMwHC}
  There is a randomized Monte Carlo algorithm solving Perfect Matching with Hitting Constraints in $\Os(2^k)$ time.
\end{restatable}

\section{From Vector Bin Packing to Exact Matching}\label{sec:PackingToMatching}

\lempacktomatch*

\begin{proof}
We interpret the problem of finding a packing as the problem of finding a~perfect matching with a certain total weight in an edge-weighted (multi-)graph. Intuitively, each large item is represented by a vertex, and an edge connects two large items if they fit together into a bin. The edge weight indicates a set of small items which can be packed together with the endpoints~(i.e., the corresponding large items). The goal is to match (pack) all large items while achieving the total weight that corresponds to all small items being assigned to some pairing of large items.

Formally, we first add $2\ell - \lvert V_L\rvert$ dummy items $\langle 0, \dots, 0 \rangle \in \Qnn^d$ to the set $V_L$ so that each bin will contain exactly two large items. 
A dummy item can be paired with another dummy item (no original large item is in that bin), or with an original large item (only one original large item is in that bin). 
For each large item $v \in V_L$ (including the dummy items), create a vertex $u_v$. 
For each pair of large items $v_1, v_2 \in V_L$, $v_1 \neq v_2$, and for each subset $V'_S \subseteq V_S$ of small items, introduce an edge between $u_{v_1}$ and $u_{v_2}$ if $v_1[i] + v_2[i] + \sum_{v \in V'_S}v[i] \leqslant 1$ for all $i \in [d]$, i.e., the small items fit together with the two large ones into a bin.\footnote{Note that it is important to add an edge for each fitting subset, and not, e.g., only for inclusion-wise maximal fitting subsets. That is because we design the edge weights so that an exact matching corresponds to a partition (and not to a cover) of the set $V_s$.} The weight of the edge will depend on $V'_S$ (but not on $v_1$ and $v_2$).

We need to design the edge weights such that each collection of edges of a certain total weight corresponds to a collection of subsets of small items that form a partition of the set of all small items $V_s$, and vice versa. A naive, but incorrect, solution would be to label the small items with integers $1, 2, \ldots, k$, and assign to a subset $X \subseteq [k]$ the integer whose binary representation corresponds to the indicator vector of $X$, i.e., $\sum_{x \in X} 2^{x-1}$. It is true that, with such weights, any collection of edges whose associated subsets form a partition of $V_s$ has the total weight $1\ldots1_2 = 2^k - 1$. However, the reverse statement is not true: it is possible to obtain the total weight $2^k-1$ by, e.g., taking $2^k-1$ edges that each allow small item $1$ but no other small items.

As we will show in Lemma~\ref{lem:edgeWeights}, in order to prevent such false positives, it suffices to concatenate the indicator vectors with $(\log k)$-bit counters denoting the number of elements in a set.\footnote{Marx and Pilipczuk~\cite{Marx14} solve a similar issue by concatenating the indicator vector with its reverse, i.e., they assign to $X$ weight $\sum_{x \in X} (2^{2k-x} + 2^{x-1})$. 
Their approach results in weights of the order of $4^k$, which is prohibitively large for achieving $\Os(2^k)$ running time.} More formally, we assign to a subset $X \subseteq [k]$ the weight $|X| \cdot 2^k + \sum_{x \in X} 2^{x-1}$, i.e., the $(k+\log k)$-bit integer whose $k$ least significant bits correspond to the indicator vector of $X$ and the $\log k$ most significant bits form the integer equal to the cardinality of $X$. 
The target total weight $k \cdot 2^k + (2^k - 1)$ can only be achieved by summing weights given to subsets forming a partition of $V_s$, i.e., by assigning each small item to (exactly) one matching edge.
\end{proof}

\begin{lemma}\label{lem:edgeWeights}
Fix the universe size $k \in \N$, and let $f: 2^{[k]} \to \N$ be given by
\[f(X) = |X| \cdot 2^k + \sum_{x \in X} 2^{x - 1}.\]
Then, a family $X_1, \ldots, X_n \subseteq [k]$ is a partition\footnote{That is, $X_1 \cup \cdots \cup X_n = [k]$, and $X_i \cap X_j = \emptyset$ for every $i \neq j$.} of $[k]$ if and only if
\[f(X_1) + \cdots + f(X_n) = k \cdot 2^k + (2^k - 1).\]
\end{lemma}

\begin{proof}
The ``partition\,$\Rightarrow$\,sum'' direction follows from a simple calculation. Let us prove the ``sum\,$\Rightarrow$\,partition'' direction.
For $i \in [k]$, let $c_i$ denote the number of sets containing element~$i$. We want to show that $c_i = 1$, for every $i$. We have
\[f(X_1) + \cdots + f(X_n) = \bigg(\sum_{i=1}^{k} c_i\bigg) \cdot 2^k + \sum_{i=1}^{k} c_i2^{i-1}.\]
Note that the $k$ least significant bits of the sum $f(X_1) + \cdots + f(X_n)$ are lower bounding the term $\sum_{i=1}^{k} c_i2^{i-1}$, and the remaining bits are upper bounding the term $\sum_{i=1}^{k} c_i$, that is,
\[
  \sum_{i=1}^{k} c_i2^{i-1} \geqslant 2^k - 1 = {\overbrace{1\ldots1}^{k \text{ ones}}}_2,
  \quad \text{and} \quad
  \sum_{i=1}^{k} c_i \leqslant k.
\]
For $i = 0, 1, \dots, k$, let $p_i = c_1 + \cdots + c_i$, with $p_0 = 0$.
Observe that $p_i \geqslant i$, for every $i$, as otherwise there are not enough bits to set the one in every position among the $i$ least significant bits of the sum $f(X_1) + \cdots + f(X_n)$.\footnote{It follows from the fact that the number of one-bits in the sum is less or equal to the total number of one-bits in the summands, and that this holds even if we look only at the $i$ least significant bits.} Moreover, $p_k = \sum_{i=1}^{k} c_i \leqslant k$, and thus $p_k = k$. Last but not least, by definition, $c_i = p_i - p_{i-1}$. We have
\begin{spreadlines}{.8em}
\begin{align*}
2^k - 1 \: & \leqslant \:
\sum_{i=1}^{k} 2^{i-1} c_i \:=\:
\sum_{i=1}^{k} 2^{i-1} (p_i - p_{i-1}) \:=\:
\sum_{i=1}^{k} 2^{i-1} p_i - \sum_{i=1}^{k} 2^{i-1} p_{i-1} \\ & = \:
\sum_{i=1}^{k} 2^{i-1} p_i - \sum_{i=0}^{k-1} 2^i p_i \:=\:
2^k p_k + \sum_{i=1}^{k} (2^{i-1}-2^i) p_i - 2^0 p_0 \:=\:
2^k p_k - \sum_{i=1}^{k} 2^{i-1} p_i \\ & = \:
2^k \cdot k - \sum_{i=1}^{k} 2^{i-1} p_i \\ & \leqslant \:
2^k \cdot k - \sum_{i=1}^{k} 2^{i-1} i \:=\:
2^k \cdot k - \big( (k-1) \cdot 2^k + 1 \big) \:=\: 2^k - 1.
\end{align*}
\end{spreadlines}
Hence, all the inequalities must be tight. In particular, $p_i = i$ for every $i$, and thus $c_i = 1$, i.e., each element of the universe is contained in exactly one set of the family.
\end{proof}

\section{Fast Exact Weight Matching in Multigraphs}

In this section we give our variant of the Mulmuley-Vazirani-Vazirani algorithm, with only a linear dependence on the edge weights and a linear additive dependence on the number of edges, proving Theorem~\ref{thm:matching}.

\subsection{The Pfaffian}
At the heart of the matching algorithm lies a computation of the Pfaffian of a skew-symmetric matrix of certain polynomials.
In order to introduce the notion of a Pfaffian properly, let us fix some definitions and notation first.

For an $n \times n$ matrix $A$, we denote by $A[i, j]$ the value in the $i$-th row and $j$-th column.
We say that $A$ is skew-symmetric if and only if $A[i,j] = -A[j, i]$ for every $i, j \in [n]$.
Let $\mathcal{M}$ be a perfect matching in the complete graph $K_{n}$.
We can look at $\mathcal{M}$ as a sequence of edges in some arbitrary order, i.e.,
\[\mathcal{M} = (i_1, j_1), (i_2, j_2), \ldots, (i_{n/2}, j_{n/2}),\]
where, by convention, $i_k \leqslant j_k$ for any $k$.
Now, we define the sign of $\mathcal{M}$ as follows:
\[\sgn \mathcal{M} = \sgn \left(
  \begin{smallmatrix}
    1 & 2 & 3 & 4 & \cdots & n-1 & n \\
    i_1 & j_1 & i_2 & j_2 & \cdots &  i_{n/2}  & j_{n/2}
  \end{smallmatrix}
\right),\]
where the right-hand side is the sign of a permutation. One can easily show that this definition does not depend on the chosen order of the edges.

Now, we are ready to give the definition of a Pfaffian.
\begin{definition}[Pfaffian]
  Let $A$ be an $n \times n$ skew-symmetric matrix.
  The Pfaffian of $A$ is denoted by $\pf(A)$ and is defined as follows
  \[\pf(A) = \sum \biggl\{ \sgn \mathcal{M} \cdot \prod_{\mathclap{(i_k, j_k) \in \mathcal{M}}} A[i_k, j_k] \biggm\vert \mathcal{M} \text{ perfect matching in } K_n \biggr\}. \]
\end{definition}
We note that since $A$ is skew-symmetric, our convention that $i_k \leqslant j_k$ does not affect the definition of the Pfaffian at all -- if we were to switch $i_k$ and $j_k$, the sign of the matching changes, but so does the sign of the product of the weights.

Several equivalent definitions of a Pfaffian exist in the literature. However, we have chosen this one, as it immediately illustrates the connection between the Pfaffian and perfect matchings.

The Pfaffian of a matrix over an arbitrary field can be computed by, e.g., a variant of the Gaussian elimination.
However, since we are dealing with polynomial matrices, we would like to avoid divisions.
Fortunately, several division-free polynomial time algorithms for computing Pfaffian exist \cite{Mahajan99, Rote2001-ra, Urbanska07}.

Incidentally, Mahajan, Subramanya, and Vinay \cite{Mahajan99} give a dynamic programming algorithm computing the Pfaffian of a matrix with entries from an arbitrary ring that makes $O(n^4)$ additions and multiplications (see also survey \cite{Rote2001-ra} for an alternative exposition)\footnote{Urbańska's algorithm \cite{Urbanska07} runs even faster, in $O(n^{3.005})$ time.
But since we care more about getting linear dependence on the target weight in our matching algorithm, rather than optimizing the polynomial dependence on $n$, we have chosen to use a slightly slower, yet simpler algorithm for the sake of clarity.}.
By analysing the structure of their algorithm, we get the following result for matrices with polynomial entries.

\begin{theorem}[cf.~\cite{Mahajan99}, Section 4] \label{alg:pfaffian}
  Given an $n \times n$ matrix $A$ of univariate polynomials of degree at most~$d$ and integer coefficients bounded by $M$,
  the Pfaffian $\pf(A)$ can be computed in $\widetilde O(n^6 d \log M)$ time.
\end{theorem}
\begin{proof}
  The algorithm in \cite{Mahajan99}, Section 4, is described as a weighted {\small DAG} $H_A$ with each vertex corresponding to a state of the dynamic program.
  The weights on the edges are signed entries of the matrix $A$.
  There is an auxiliary starting state $s \in H_A$ and the dynamic programming value for a state $v \in H_A$ is a sum of products of weights along all the paths from $s$ to $v$.

  Moreover, $H_A$ has $O(n^3)$ vertices, depth equal to $O(n)$ and indegree of each vertex equal to $O(n)$.
  Therefore, if the entries of $A$ are polynomials of degree $d$ and coefficients bounded by $M$, then the values of the dynamic programming states are polynomials with a degree bounded by $O(nd)$ and coefficients bounded by $O(n^n M^n)$.
  Hence, by using {\small FFT}, we can perform each arithmetic operation in $\widetilde O(n^2 d \log M)$ time.
  The number of arithmetic operations needed is proportional to the number of edges in $H_A$, which is $O(n^4)$.
  This yields the desired time bound.
\end{proof}

Since we do not need to compute the whole Pfaffian in the Exact Matching problem, but are only interested in the coefficient of the monomial $x^t$ (which conveys the information about matchings of the target weight $t$), we can speed up the computation by a factor of $n$.

\begin{corollary} \label{alg:pfaffian-coefficient}
Given an integer $t$ and an $n \times n$ matrix $A$ of univariate polynomials with integer coefficients bounded by $M$, a coefficient of the monomial $x^t$ in $\pf(A)$ can be computed in $\widetilde O(n^5 t \log M)$ time.
\end{corollary}
\begin{proof}
In the algorithm from Theorem \ref{alg:pfaffian}, we can perform all the arithmetic operations modulo $x^{t+1}$. Then, the degree of the polynomials is bounded by $O(t)$ instead of $O(nd)$, and a similar analysis follows.
\end{proof}

\subsection{The algorithm}

We first recall the central lemma of the Mulmuley-Vazirani-Vazirani algorithm, used to deal with possible cancellations caused by varying signs in the Pfaffian definition.

\begin{lemma}[Isolation Lemma, cf.~\cite{Mulmuley1987-mc}] \label{lem:isolation-lemma}
  Let $S$ be a finite set, and let $F \subseteq 2^S$ be a family of subsets of $S$.
  To each element $x \in S$, we assign an integer cost $c(x)$ chosen uniformly and independently at random from $\{1, \ldots, 2 |S|\}$.
  For a subset $S' \subseteq S$, we define a total cost of $S'$ to be $c(S') = \sum_{x \in S'} c(x)$.
  Then,
  \[\mathbb{P}(\text{there is a unique minimum total cost set in } F) \geqslant \frac{1}{2}.\]
\end{lemma}

\noindent
Now we are ready to present the matching algorithm.

\thmmatching*

\begin{proof} We first present the algorithm. Then we argue its correctness and analyse the running time.  

   \vskip 1em \noindent\textbf{Algorithm.}\quad     
   For every $\{u, v\} \in \binom{V}{2}$, let $E_{\{u,v\}} = \{ e \in E : e \text{ connects $u$ and $v$}\}$ denote the set of (parallel) edges between nodes $u$ and $v$.
   For an edge $e \in E$, we use $w(e) \in \mathbb{Z}_{\geqslant 0}$ to denote the weight of $e$.
   Moreover, we assume w.l.o.g.~that $V = [n]$.

   The algorithm works as follows.
   
    \begin{enumerate}
      \item Set $\lambda = 2m^{n}$.
      \item For every $\{i,j\} \in \binom{V}{2}$, assign a cost $c(\{i,j\})$ uniformly at random from $\{1, \ldots, 2\binom{n}{2} \}$.
      \item Set up an $n \times n$ matrix $A$ of univariate polynomials:
      For each $i, j \in [n]$, $i \leqslant j$, put
      \[A[i,j] = \lambda^{c(\{i,j\})} \sum_{\mathclap{e \in E_{\{i,j\}}}} x^{w(e)}, \quad \text{and} \quad A[j,i] = -A[i,j].\]
      \item Compute the coefficient of $x^t$ in $\pf(A)$ using the algorithm from Corollary \ref{alg:pfaffian-coefficient}.
      \item If the coefficient of $x^t$ in $\pf(A)$ is nonzero return {\small YES}, otherwise return {\small NO}.
    \end{enumerate}

  \vskip 1em \noindent\textbf{Correctness.}\quad     
  We use $\coef_{t}(\pf(A))$ to denote the coefficient of $x^t$ in $\pf(A)$.
  For every perfect matching $\mathcal{M}$ in the complete graph $K_n$, let
  \[
     f(\mathcal{M}) = \sgn~\mathcal{M} \cdot \lambda^{c(\mathcal{M})} \cdot \, \# \text{perfect matchings in } G \text{ of weight } t \text{ contained\footnotemark{} in } \mathcal{M}
  \]
  \footnotetext{We say that a matching (in multigraph $G$) is contained in another matching (in the complete graph~$K_n$) if the set of $n/2$ pairs of endpoints is the same for the two matchings.}
  denote the contribution of matching $\mathcal{M}$ to the coefficient $\coef_t(\pf(A))$.
  Now, we have
  \begin{equation}\label{eq:pft}
    \coef_t(\pf(A)) = \sum \bigl\{ f(\mathcal{M}) \bigm\vert \mathcal{M} \text{ perfect matching in } K_n \bigr\}.
  \end{equation}

  Let $F$ be the family of all perfect matchings in $K_n$ that contain a perfect matching in $G$ of weight exactly $t$.
  If $F = \emptyset$, then every summand in (\ref{eq:pft}) is zero. Hence, $\coef_t(\pf(A)) = 0$ and our algorithm answers correctly.

  If $F \neq \emptyset$, then by Isolation Lemma, with probability at least $\nicefrac{1}{2}$, there is only one minimum cost perfect matching $\mathcal{N} \in F$.

  Let $c = c(\mathcal{N})$.
  Observe that the number of perfect matchings in $G$ of weight $t$ that are contained in $\mathcal{N}$ is trivially bounded by $m^n < \lambda$.
  This means that $|f(\mathcal{N})| < \lambda^{c+1}$. In other words, $f(\mathcal{N})$ is divisible by $\lambda^c$, but not by $\lambda^{c+1}$.
  On the other hand, every other summand in~(\ref{eq:pft}) is divisible by $\lambda^{c+1}$, as $\mathcal{N}$ is the unique minimum cost matching.
  Therefore, $\coef_t(\pf(A))$ is divisible by $\lambda^c$, but not by $\lambda^{c+1}$ -- so it cannot be zero.

  If we want to amplify the probability of giving the correct answer to $1 - \nicefrac{1}{n^C}$, for some constant $C > 0$, we repeat the algorithm $C\log n$ times.

  \vskip 1em \noindent\textbf{Time cost analysis.}\quad     
  The time needed to complete steps 1--3 is $O(n^2 + m)$.
  Since the coefficients of the polynomial entries of $A$ are bounded by $2m^{n\cdot2\binom{n}{2}} = 2^{O(n^3 \log m)}$, we get that invoking the algorithm from Corollary \ref{alg:pfaffian-coefficient} takes $\widetilde O(t \cdot n^8 \log m)$ time.
  In total, that yields $\widetilde O(t \cdot n^8 + m)$ time complexity.
\end{proof}

\section{Lower bound}

\thmlowerbound*

\begin{proof}
Given a {\small CNF} formula with $n$ variables and $m$ clauses,\footnote{Note that, thanks to the \emph{sparsification lemma}~\cite{Impagliazzo98}, we can assume that $m = O(n)$.} we will construct $n+1$ instances of Vector Bin Packing such that the formula is satisfiable if and only if at least one of them is a yes-instance. Intuitively, this corresponds to guessing the number of variables set to true in a satisfying assignment. Formally, for $t \in \{0,\ldots,n\}$, the $t$-th Vector Bin Packing instance is a yes-instance if and only if the formula has a satisfying assignment with exactly $t$ variables set to true.

Let us fix $t$. The $t$-th instance consists of $n+2$ items $V = \{v_1, \ldots, v_n, T, F\} \subseteq \Qnn^{m+2}$. The first $n$ items correspond to the $n$ variables; the remaining two items $T$ and $F$ are used to break the symmetry -- in any feasible solution they are necessarily in two different bins, which we call the $T$-bin and the $F$-bin, respectively. Packing item $v_i$ to the $T$-bin corresponds to setting variable $i$ to true, and packing it to the $F$-bin corresponds to setting the variable to false.

The items are $(m+2)$-dimensional. The first $m$ dimensions correspond to clauses, and we will discuss them in a moment. Dimension $m+1$ ensures that $T$ and $F$ go to different bins; we have
$T[m+1]=F[m+1]=1$, and $v_i[m+1] = 0$ for every $i \in [n]$. Dimension $m+2$ ensures that (at most) $t$ items go to the $T$-bin and (at most) $n-t$ items go to the $F$-bin; we have $T[m+2] = (n-t)/n$, $F[m+2]=t/n$, and $v_i[m+2]=1/n$ for every $i \in [n]$.

Now, fix a clause $j \in [m]$. We set
\[v_i[j] = \begin{cases}
  \nicefrac{0}{2n}, & \text{if variable $i$ appears in a positive literal in clause $j$}, \\
  \nicefrac{1}{2n}, & \text{if variable $i$ does not appear in clause $j$}, \\
  \nicefrac{2}{2n}, & \text{if variable $i$ appears in a negative literal in clause $j$}.
\end{cases}\]
Let $n_j$ denote the number of variables that appear negated in clause $j$. We set
\[T[j] = 1 - \frac{t + n_j - 1}{2n}, \quad \text{and} \quad F[j] = 0.\]

This ends the description of the instance. To finish the proof, it remains to show that the above items can be packed into two bins if and only if the formula has a~satisfying assignment with exactly $t$ variables set to true.

Note that there is a natural one-to-one correspondence between (not necessarily satisfying) assignments that set exactly $t$ variables to true and (not necessarily feasible) Vector Bin Packing solutions that are feasible in the last two dimensions. We now show that, for $j \in [m]$, such an assignment satisfies clause $j$ if and only if the corresponding solution is feasible in dimension~$j$. The $F$-bin is never overfull in dimension $j$. To analyse the $T$-bin, let $\alpha$, $\beta$, $\gamma$ denote the numbers of variables set to true that, in clause $j$, appear in a positive literal, do not appear, and appear in a negative literal, respectively. Let $\delta$ denote the number of variables set to false that appear in clause $j$ in a negative literal. Note that $t = \alpha + \beta + \gamma$, and $n_j = \gamma + \delta$. Consider the following chain of equivalent inequalities, starting with the condition saying that the $T$-bin is not overfull in dimension $j$.
\begin{align*}
\alpha \cdot \nicefrac{0}{2n} + \beta \cdot \nicefrac{1}{2n} + \gamma \cdot \nicefrac{2}{2n} & \leqslant 1 - T[j] \\
\beta + 2\gamma & \leqslant t + n_j - 1 \\
\beta + 2\gamma & \leqslant (\alpha + \beta + \gamma) + (\gamma + \delta) - 1 \\
1 & \leqslant \alpha + \delta
\end{align*}
The last inequality states that clause $j$ is satisfied.
\end{proof}

\section{Other applications}
In this section we explain how the techniques presented in our paper can be adapted to also solve Vector Multiple Knapsack and Vector Bin Covering, two closely related problems to the Vector Bin Packing problem. The main difference lies in the reduction to the Exact Matching problem, which has to integrate profits of the items, or the new covering property, respectively. Further, we show that our techniques directly apply to the Perfect Matching with Hitting Constraints problem, leading to an improved running time.

\subsubsection*{Vector Multiple Knapsack}
In Vector Multiple Knapsack, instead of packing all items into the smallest number of bins, we aim to place a subset of items with profits into a fixed number of bins while maximizing the profit of the packed items. Like in Vector Bin Packing, small items hinder us from solving the problem using a polynomial time algorithm for the maximum weight perfect matching. Hence, following Bannach et al.~\cite{Bannach20}, we study the problem parameterized by the number $k$ of small items.

\begin{problem}{Vector Multiple Knapsack with Few Small Items}
  Parameter:
    & the number of \emph{small} items $k$. \\[3pt]
  Given:
    & a set of $n$ items $V = \{v_1, \ldots, v_n\} \subseteq \Qnn^{d}$, \textbf{item profits} $p(v_1), \dots, p(v_n) \in \mathbb{Z}_{+}$,\\
    & a subset of $k$ items $V_S \subseteq V$ such that $V_L = V \setminus V_S$ is \textbf{3-incompatible},\\
    & an integer $\ell \in \mathbb{Z}_+$ denoting the number of unit-sized bins, \\
    & and an integer $P \in \mathbb{Z}_{+}$, denoting the \textbf{goal profit}.\\[3pt]
  Decide:
    & if a subset $V'$ of the items can be partitioned into $\ell$ bins $B_1 \cup \cdots \cup B_\ell = V'$  \\
    & such that $\sum_{v \in B_i} v[j] \leqslant 1$ for every bin $i \in [\ell]$ and every dimension $j \in [d]$, \\
    & and $\sum_{v \in V'} p(v) \geqslant P$.
\end{problem}

To solve the problem, we reduce the instance to the Exact Matching problem as in Section~\ref{sec:PackingToMatching}.
It remains to handle the fact that only a subset of items has to be packed, and that we need to integrate the profits. 
We do so in the following manner: With each edge between $v_1$ and $v_2$ and the weight corresponding to $V'_S \subseteq V_S$, we associate the \emph{cost} of $p(v_1)+p(v_2)+\sum_{v \in V'_s} p(v)$.
Further, we introduce $g = n - 2 \cdot \ell$ new vertices $b_1, b_2 \dots, b_g$, called \emph{blocker} vertices. These vertices serve as ``garbage collectors'' for the items which are not packed in any of the $\ell$ bins, i.e., they match $g$ unpacked items, and by that block them. To do so, for each $V'_S \subseteq V_S$, each large vector $v_i$, and each blocker vertex $b_j$, introduce an edge between $v_i$ and $b_j$ with weight dependent on $V'_S$ as before, and cost $0$. Note that, because of the dummy items introduced in the reduction in Section~\ref{sec:PackingToMatching}, we can assume that each bin in an optimal solution contains exactly two large items (some original, some dummy), so we know that exactly $g = n - 2 \cdot \ell$ large items has to be handled by blockers.

Using Lemma~\ref{lem:edgeWeights}, clearly, each yes-instance of the Vector Multiple Knapsack problem has a prefect matching of weight exactly $k \cdot 2^k + (2^k -1)$ and cost at least $P$ in the above graph, and vice versa. This is due to the equivalence of choosing $\ell$ edges with non-zero costs and the packing of the $\ell$ bins. The remaining items can be matched with the blocker vertices, and all small items are covered due to the weights.

We are left with solving the following matching problem: Given a (multi-)graph with edge weight and edge costs, find a perfect matching with a given total weight and the maximum possible total cost. This can be done with a slight modification of the algorithm of Theorem~\ref{thm:matching}. Indeed, note that the algorithm already looks for a perfect matching minimizing the sum of edge costs coming from Isolation Lemma. All we have to do is to (1) combine input costs with Isolation Lemma costs, and (2) turn minimization into maximization. For (1), it suffices to put the input cost into the most significant bits, and the Isolation Lemma cost into the least significant bits of the combined edge cost. For (2), to find out what the maximum (instead of the minimum) possible total cost is, it suffices to look at the most (instead of the least) significant digit in the $\lambda$-ary representation of the coefficient $\coef_t(\pf(A))$. Last but not least, we remark that Isolation Lemma is symmetric with respect to minimization/maximization, i.e., it also ensures that the maximum total cost set is unique with probability at least $\nicefrac{1}{2}$.

To analyze the running time, let $p_{\max} = \max_{v \in V} p(v)$ denote the maximum item profit. The maximum input cost of an edge is $(k+2) p_{\max}$. Hence, the coefficients of the polynomial entries of matrix $A$ are now bounded by $2m^{(k+2)p_{\max}n\cdot2\binom{n}{2}} = 2^{O(p_{\max}n^4 \log m)}$, and the matching algorithm takes $\widetilde O(t \cdot p_{\max} n^9 \log m)$ time. This leads to the following theorem.

\thmKnapsack*

\subsubsection*{Vector Bin Covering}
Another set of problems asks to \emph{cover} the largest number of bins possible. In the one-dimensional setting, covering typically refers to the bin capacity being exceeded by the set of items packed into it. This property can be extended in multiple ways to a~$d$-dimensional case, for example by requiring that at least one dimension is exceeded. However, other properties, such as ``all dimension have to be exceeded'', ``certain set combinations of dimensions have to be exceeded'', et cetera, are possible as well. Our algorithm works for all such definitions of covering. Thus, in the following, we refer to the one chosen as the covering property $\mathcal P$.

Following our story line to study a parameter capturing a distance to triviality, we consider the problem variant parameterized by the number $k$ of small items. However, the property of being a small item depends on $\mathcal P$, so we introduce a new definition for the covering problems: We say that a subset $V' \subseteq V$ is \textbf{3-covering} if every three distinct items from $V'$ cover a unit-sized bin w.r.t.\ $\mathcal P$.

\begin{problem}{Vector Bin Covering with Few Small Items}
  Parameter:
    & the number of \emph{small} items $k$. \\[3pt]
  Given:
    & a set of $n$ items $V = \{v_1, \ldots, v_n\} \subseteq \Qnn^d$, \\
    & a subset of $k$ items $V_S \subseteq V$ such that $V_L = V \setminus V_S$ is \textbf{3-covering} w.r.t.~$\mathcal P$, \\
    & and an integer $\ell \in \mathbb{Z}_+$ denoting the number of unit-sized bins. \\[3pt]
  Decide:
    & if the items can be partitioned into $\ell$ bins $B_1 \cup \cdots \cup B_\ell = V$ such that \\
    & $\sum_{v \in B_i} v$ satisfies $\mathcal P$ for every bin $i \in [\ell]$.
\end{problem}

The algorithm proceeds similarly to the one for Vector Bin Packing. However, we have to handle the fact that a bin can contain more than two large items in this case. Thus, we first guess the number of bins $\ell_i$ admitting $i$ large items for $i \in \{0,1,2\}$. This yields $O(\ell^3) = O(n^3)$ guesses. The remaining bins will be covered by triples of the unassigned large items. Hence, the guess has to satisfy that $\ell_0+\ell_1+\ell_2+\lfloor(n-k-\ell_1-2\ell_2)/3\rfloor \geqslant \ell$.

Now we construct the graph as in Section~\ref{sec:PackingToMatching} with $2\ell_0 + \ell_1$ dummy items. For each $V'_S \subseteq V_S$, an edge is introduced between $v_1$ and $v_2$ if $v_1 + v_2 +\sum_{v \in V'_S} v$ covers the bin w.r.t.\ $\mathcal P$. The weight of the edge is defined by $V'_S$ as before. Additionally, we introduce $(n-k-\ell_1-2\ell_2)$ \emph{blocker} vertices, and introduce an edge of weight $0$ between each blocker vertex and each large item. The blocker vertices collect all large items not being placed into bins with 0, 1, or 2 large items.

With Lemma~\ref{lem:edgeWeights} being proven, clearly, each yes-instance of the Vector Bin Covering problem has a perfect matching of weight $k \cdot 2^k + (2^k -1)$ in the above graph, and vice versa. Indeed, a matching has to choose $(n-k-\ell_1-2\ell_2)$ edges between blocker vertices and large items. These are the ones greedily packed as triples. Note that this might leave up to two large items unpacked, which will be assigned to an arbitrary, already covered bin. 
The remaining packing is defined by the remaining matching edges as previously. 

This together with Theorem~\ref{thm:matching} leads to the following result.

\thmCover*

\subsubsection*{Perfect Matching with Hitting Constraints}
The Perfect Matching with Hitting Constraints problem asks whether there exists a perfect matching in a graph using at least one edge from each given set of edges. Formally, the problem is defined as follows.

\begin{problem}{Perfect Matching with Hitting Constraints}
  Parameter:
    & the number of edge subsets $k$. \\[3pt]
  Given:
    & a graph $G = \langle V, E \rangle$, \\
    & and $k$ (not necessarily disjoint) edge subsets $E_1,\dots, E_k \subseteq E$.\\[3pt]
  Decide:
    & if there is a perfect matching $M$ in $G$ such that \\
    & there exists $k$ \emph{distinct} edges $e_1, \dots, e_k \in M$ such that $e_i \in E_i$ for every $i \in [k]$.
\end{problem}

We again reduce this problem to finding an exact weight perfect matching in a~multigraph. Our approach is similar to the one of Marx and Pilipczuk~\cite{Marx14}. However, in their reduction, they introduce larger edge weights, and, by that, obtain a larger running time. We can circumvent this using edge weights as defined in Lemma~\ref{lem:edgeWeights}.

In detail, we create a copy of each edge $e \in E_i$, for each $i \in [k]$, and assign weight $1 \cdot 2^k + 2^{i-1}$ to it -- i.e., we concatenate the indicator vector of the singleton $\{i\}$ with the counter set to $1$, as previously. The original edge gets weight $0$. The target weight is $t=k \cdot 2^k + (2^k -1)$. Clearly, there exists a perfect matching with hitting constraints in $G$ if and only if there is a perfect matching in the transformed graph with edge weights summing up to the correct target value~$t$, see Lemma~\ref{lem:edgeWeights}.

This together with Theorem~\ref{thm:matching} leads to the following result.

\thmPMwHC*

\bibliography{main}

\begin{thebibliography}{10}

\bibitem{DBLP:books/aw/AhoHU74}
Alfred~V. Aho, John~E. Hopcroft, and Jeffrey~D. Ullman.
\newblock {\em The Design and Analysis of Computer Algorithms}.
\newblock Addison-Wesley, 1974.

\bibitem{AlberBFKN02}
Jochen Alber, Hans~L. Bodlaender, Henning Fernau, Ton Kloks, and Rolf
  Niedermeier.
\newblock Fixed parameter algorithms for {DOMINATING} {SET} and related
  problems on planar graphs.
\newblock {\em Algorithmica}, 33(4):461--493, 2002.
\newblock \href {https://doi.org/10.1007/s00453-001-0116-5}
  {\path{doi:10.1007/s00453-001-0116-5}}.

\bibitem{DBLP:conf/soda/AlmanW21}
Josh Alman and Virginia~Vassilevska Williams.
\newblock A refined laser method and faster matrix multiplication.
\newblock In {\em Proceedings of the 2021 {ACM-SIAM} Symposium on Discrete
  Algorithms, {SODA} 2021, Virtual Conference, January 10 - 13, 2021}, pages
  522--539. {SIAM}, 2021.

\bibitem{Bannach20}
Max Bannach, Sebastian Berndt, Marten Maack, Matthias Mnich, Alexandra Lassota,
  Malin Rau, and Malte Skambath.
\newblock {Solving Packing Problems with Few Small Items Using Rainbow
  Matchings}.
\newblock In {\em {MFCS}}, volume 170 of {\em Leibniz International Proceedings
  in Informatics (LIPIcs)}, pages 11:1--11:14. Schloss
  Dagstuhl--Leibniz-Zentrum f{\"u}r Informatik, 2020.
\newblock \href {https://doi.org/10.4230/LIPIcs.MFCS.2020.11}
  {\path{doi:10.4230/LIPIcs.MFCS.2020.11}}.

\bibitem{Bannach20Long}
Max Bannach, Sebastian Berndt, Marten Maack, Matthias Mnich, Alexandra Lassota,
  Malin Rau, and Malte Skambath.
\newblock {Solving Packing Problems with Few Small Items Using Rainbow
  Matchings}.
\newblock {\em CoRR}, abs/2007.02660, 2020.
\newblock URL: \url{https://arxiv.org/abs/2007.02660v1}.

\bibitem{CyganFKLMPPS15}
Marek Cygan, Fedor~V. Fomin, Lukasz Kowalik, Daniel Lokshtanov, D{\'{a}}niel
  Marx, Marcin Pilipczuk, Michal Pilipczuk, and Saket Saurabh.
\newblock {\em Parameterized Algorithms}.
\newblock Springer, 2015.
\newblock \href {https://doi.org/10.1007/978-3-319-21275-3}
  {\path{doi:10.1007/978-3-319-21275-3}}.

\bibitem{Edmonds65}
Jack Edmonds.
\newblock Paths, trees, and flowers.
\newblock {\em Canadian Journal of Mathematics}, 17:449--467, 1965.
\newblock \href {https://doi.org/10.4153/CJM-1965-045-4}
  {\path{doi:10.4153/CJM-1965-045-4}}.

\bibitem{GuoHN04}
Jiong Guo, Falk H{\"{u}}ffner, and Rolf Niedermeier.
\newblock A structural view on parameterizing problems: Distance from
  triviality.
\newblock In {\em {IWPEC}}, volume 3162 of {\em Lecture Notes in Computer
  Science}, pages 162--173. Springer, 2004.
\newblock \href {https://doi.org/10.1007/978-3-540-28639-4\_15}
  {\path{doi:10.1007/978-3-540-28639-4\_15}}.

\bibitem{DBLP:journals/jcss/GutinWY17}
Gregory~Z. Gutin, Magnus Wahlstr{\"{o}}m, and Anders Yeo.
\newblock Rural postman parameterized by the number of components of required
  edges.
\newblock {\em J. Comput. Syst. Sci.}, 83(1):121--131, 2017.
\newblock \href {https://doi.org/10.1016/j.jcss.2016.06.001}
  {\path{doi:10.1016/j.jcss.2016.06.001}}.

\bibitem{ImpagliazzoP01}
Russell Impagliazzo and Ramamohan Paturi.
\newblock On the complexity of k-{SAT}.
\newblock {\em Journal of Computer and System Sciences}, 62(2):367--375, 2001.
\newblock \href {https://doi.org/10.1006/jcss.2000.1727}
  {\path{doi:10.1006/jcss.2000.1727}}.

\bibitem{Impagliazzo98}
Russell Impagliazzo, Ramamohan Paturi, and Francis Zane.
\newblock Which problems have strongly exponential complexity?
\newblock In {\em {FOCS}}, pages 653--663. {IEEE} Computer Society, 1998.
\newblock \href {https://doi.org/10.1109/SFCS.1998.743516}
  {\path{doi:10.1109/SFCS.1998.743516}}.

\bibitem{JansenLL16}
Klaus Jansen, Felix Land, and Kati Land.
\newblock Bounding the running time of algorithms for scheduling and packing
  problems.
\newblock {\em {SIAM} Journal on Discrete Mathematics}, 30(1):343--366, 2016.
\newblock \href {https://doi.org/10.1137/140952636}
  {\path{doi:10.1137/140952636}}.

\bibitem{Mahajan99}
Meena Mahajan, P.~R. Subramanya, and V.~Vinay.
\newblock A combinatorial algorithm for {Pfaffians}.
\newblock In {\em {COCOON}}, COCOON'99, page 134–143, Berlin, Heidelberg,
  1999. Springer-Verlag.
\newblock \href {https://doi.org/10.1007/3-540-48686-0\_13}
  {\path{doi:10.1007/3-540-48686-0\_13}}.

\bibitem{Marx14}
D{\'{a}}niel Marx and Michal Pilipczuk.
\newblock Everything you always wanted to know about the parameterized
  complexity of subgraph isomorphism (but were afraid to ask).
\newblock In {\em {STACS}}, volume~25 of {\em LIPIcs}, pages 542--553. Schloss
  Dagstuhl - Leibniz-Zentrum f{\"{u}}r Informatik, 2014.
\newblock \href {https://doi.org/10.4230/LIPIcs.STACS.2014.542}
  {\path{doi:10.4230/LIPIcs.STACS.2014.542}}.

\bibitem{Mulmuley1987-mc}
Ketan Mulmuley, Umesh~V Vazirani, and Vijay~V Vazirani.
\newblock Matching is as easy as matrix inversion.
\newblock {\em Combinatorica}, 7(1):105--113, March 1987.
\newblock \href {https://doi.org/10.1007/BF02579206}
  {\path{doi:10.1007/BF02579206}}.

\bibitem{NederlofPSW21}
Jesper Nederlof, Jakub Pawlewicz, C{\'{e}}line M.~F. Swennenhuis, and Karol
  Wegrzycki.
\newblock A faster exponential time algorithm for bin packing with a constant
  number of bins via additive combinatorics.
\newblock In {\em {SODA}}, pages 1682--1701. {SIAM}, 2021.
\newblock \href {https://doi.org/10.1137/1.9781611976465.102}
  {\path{doi:10.1137/1.9781611976465.102}}.

\bibitem{Papadimitriou1982-qg}
Christos~H Papadimitriou and Mihalis Yannakakis.
\newblock The complexity of restricted spanning tree problems.
\newblock {\em Journal of the {ACM}}, 29(2):285--309, 1982.
\newblock \href {https://doi.org/10.1145/322307.322309}
  {\path{doi:10.1145/322307.322309}}.

\bibitem{RazgonO08}
Igor Razgon and Barry O'Sullivan.
\newblock Almost {2-SAT} is fixed-parameter tractable (extended abstract).
\newblock In {\em {ICALP}}, volume 5125 of {\em Lecture Notes in Computer
  Science}, pages 551--562. Springer, 2008.
\newblock \href {https://doi.org/10.1007/978-3-540-70575-8\_45}
  {\path{doi:10.1007/978-3-540-70575-8\_45}}.

\bibitem{Rote2001-ra}
G{\"u}nter Rote.
\newblock Division-free algorithms for the determinant and the {Pfaffian}:
  algebraic and combinatorial approaches.
\newblock In {\em Computational Discrete Mathematics: advanced lectures}, pages
  119--135. Springer-Verlag, 2001.
\newblock \href {https://doi.org/10.1007/3-540-45506-X\_9}
  {\path{doi:10.1007/3-540-45506-X\_9}}.

\bibitem{DBLP:journals/jacm/Schwartz80}
Jacob~T. Schwartz.
\newblock Fast probabilistic algorithms for verification of polynomial
  identities.
\newblock {\em J. {ACM}}, 27(4):701--717, 1980.
\newblock \href {https://doi.org/10.1145/322217.322225}
  {\path{doi:10.1145/322217.322225}}.

\bibitem{SvenssonT17}
Ola Svensson and Jakub Tarnawski.
\newblock The matching problem in general graphs is in quasi-{NC}.
\newblock In {\em {FOCS}}, pages 696--707. {IEEE} Computer Society, 2017.
\newblock \href {https://doi.org/10.1109/FOCS.2017.70}
  {\path{doi:10.1109/FOCS.2017.70}}.

\bibitem{Urbanska07}
Anna Urba{\'{n}}ska.
\newblock Faster combinatorial algorithms for determinant and {Pfaffian}.
\newblock In {\em Algorithms and Computation, 18th International Symposium,
  {ISAAC} 2007}, pages 599--608. Springer Berlin Heidelberg, 2007.
\newblock \href {https://doi.org/10.1007/978-3-540-77120-3\_52}
  {\path{doi:10.1007/978-3-540-77120-3\_52}}.

\bibitem{Wahlstrom07}
Magnus Wahlstr{\"{o}}m.
\newblock {\em Algorithms, measures and upper bounds for satisfiability and
  related problems}.
\newblock PhD thesis, Link{\"{o}}ping University, Sweden, 2007.
\newblock URL: \url{https://nbn-resolving.org/urn:nbn:se:liu:diva-8714}.

\bibitem{DBLP:conf/stacs/Wahlstrom13}
Magnus Wahlstr{\"{o}}m.
\newblock Abusing the tutte matrix: An algebraic instance compression for the
  k-set-cycle problem.
\newblock In {\em 30th International Symposium on Theoretical Aspects of
  Computer Science, {STACS} 2013}, volume~20 of {\em LIPIcs}, pages 341--352.
  Schloss Dagstuhl - Leibniz-Zentrum f{\"{u}}r Informatik, 2013.
\newblock \href {https://doi.org/10.4230/LIPIcs.STACS.2013.341}
  {\path{doi:10.4230/LIPIcs.STACS.2013.341}}.

\bibitem{williams2018some}
Virginia~Vassilevska Williams.
\newblock On some fine-grained questions in algorithms and complexity.
\newblock In {\em {ICM}}, pages 3447--3487. World Scientific, 2018.
\newblock \href {https://doi.org/10.1142/9789813272880_0188}
  {\path{doi:10.1142/9789813272880_0188}}.

\bibitem{yates1937design}
Frank Yates.
\newblock The design and analysis of factorial experiments.
\newblock 1937.

\bibitem{DBLP:conf/eurosam/Zippel79}
Richard Zippel.
\newblock Probabilistic algorithms for sparse polynomials.
\newblock In {\em Symbolic and Algebraic Computation, {EUROSAM} '79, An
  International Symposiumon Symbolic and Algebraic Computation}, volume~72 of
  {\em Lecture Notes in Computer Science}, pages 216--226. Springer, 1979.
\newblock \href {https://doi.org/10.1007/3-540-09519-5\_73}
  {\path{doi:10.1007/3-540-09519-5\_73}}.

\end{thebibliography}

\clearpage
\appendix

\section{Alternative algorithm for Vector Bin Packing}
\label{sec:alt}

In this section, we develop an alternative algorithm for Vector Bin Packing with a better dependence on the number of items $n$ than the algorithm from Theorem~\ref{thm:binpacking}.
We follow the approach of Gutin et al.~\cite{DBLP:journals/jcss/GutinWY17} for the Conjoining Bipartite Matching problem that is based on the following lemma by Wahlstr{\"{o}}m \cite{DBLP:conf/stacs/Wahlstrom13}.

\begin{lemma}[cf. \cite{DBLP:conf/stacs/Wahlstrom13}, Lemma 2] \label{lem:Qpolynomial}
Let $P(x_1, \ldots, x_n)$ be a polynomial over a field of characteristic $2$. 
For a set $I \subseteq \{x_1,\ldots,x_n\}$, define $P_{-I}(x_1, \ldots, x_n) = P(y_1, \ldots, y_n)$ where $y_i = 0$ if $x_i \in I$, and $y_i = x_i$ otherwise.
For a $J \subseteq \{x_1,\ldots,x_n\}$ define
$$
\Phi_J(P) = \sum\limits_{I \subseteq J} P_{-I}.
$$
For a monomial $T$ and a polynomial $Q$ let us use\,\ $\coef_T(Q)$ to denote the coefficient of $T$ in $Q$.
Then, for any monomial $T$ we have
\[\coef_T(\Phi_J(P)) = \begin{cases}
  \coef_T(P), & \text{if~} T \text{~is divisible by~} \prod_{x_i \in J} x_i, \\
  0, & \text{otherwise}.
\end{cases}\]
\end{lemma}

We are also going to use the classic Schwartz-Zippel lemma:

\begin{lemma}[Schwartz-Zippel, \cite{DBLP:journals/jacm/Schwartz80, DBLP:conf/eurosam/Zippel79}]
Let $P(x_1, \ldots, x_n)$ be a multivariate polynomial of maximum degree at most $d$ over a field\:\ $\mathbb{F}$, and assume that $P$ is not identically equal to zero.
Pick $r_1, \ldots, r_n$ uniformly at random from $\mathbb{F}$.
Then $\mathbb{P}(P(r_1, \ldots, r_n) = 0) \leqslant d/|\mathbb{F}|$.
\end{lemma}
 
\begin{theorem}
There is a randomized Monte Carlo algorithm solving Vector Bin Packing
with Few Small Items in $O(2^k n^{\omega + o(1)} + 2^k k n^2\operatorname{polylog}(n))$ time.
\end{theorem} 

\begin{proof}
Recall the (multi-)graph that we construct in the proof of Lemma~\ref{lem:packtomatch}. 
Let us call it $G$, and assume w.l.o.g.~that the vertex set of $G$ is $[N]$, for $N = O(n)$. Recall that each of the $M = O(2^k n^2)$ edges of $G$ represented a subset of $[k]$. For the purpose of the current proof, we modify the edge weights -- in Lemma \ref{lem:packtomatch}, we set integer weights, but now, to every edge of $G$, we assign a weight that is an appropriately chosen polynomial over a field of characteristic $2$.

More specifically, for every $i \in [k]$, we create a variable $x_i$.
Moreover, for every edge $e$, we create an auxiliary variable $z_e$.
Now, for every pair of vertices $i, j \in [N]$ and every subset $J \subseteq [k]$, we define
\[P_{i, j, J} = \begin{cases}
z_e\prod\limits_{\mathclap{s \in J}} x_s, & \text{if between vertices $i$ and $j$ there is edge $e$ corresponding to $J$}, \\
  
  0, & \text{otherwise}.
\end{cases}\]
Note that we interpret $P_{i, j, J}$ as a polynomial over a finite field $\mathbb{F}$ of characteristic $2$, with the exact size of $\mathbb{F}$ to be determined later.

Next, we construct a matrix $A$ of polynomials over $\mathbb{F}$ in such a way that the Pfaffian of $A$ conveys information about existence of a solution for our Vector Bin Packing instance.
For each $i, j \in [N]$, with $i \leqslant j$, let us put
\[A[i,j] = \sum_{\mathclap{J \subseteq [k]}}  P_{i, j, J},
    \quad \text{and} \quad A[j,i] = -A[i,j].\]

Observe that, thanks to the introduction of $z$-variables, there are no term cancellations in the Pfaffian $\pf(A)$. Therefore, a solution to the Vector Bin Packing instance exists if and only if there exists a monomial in $\pf(A)$ divisible by $\prod_{i=1}^{k} x_i$. By Lemma~\ref{lem:Qpolynomial}, this is equivalent to $\Phi_X(\pf(A))$, for $X = \{x_1,\ldots,x_k\}$, being not identically equal to zero.
That, on the other hand, can be checked using Schwartz-Zippel lemma.
In order to use that lemma, we would like to show how to evaluate $\Phi_X(\pf(A))$ efficiently.

Observe that $\deg \Phi_X(\pf(A)) \leqslant (k+1)n$.
Hence, let us fix $\mathbb{F}$ to be the field of size $2^q$ for $q = \lceil\log(2(k+1)n)\rceil$.
Now, let us choose values $r_1, \ldots, r_k$ and $s_1, \ldots, s_M$ (for $x$-variables and $z$-variables, respectively) uniformly at random from $\mathbb{F}$.
Let us refer to these two vectors of numbers as $\Bar{r}$ and $\Bar{s}$ for brevity.
By Schwartz-Zippel lemma, if $\Phi_X(\pf(A)) \not\equiv 0$, then $\mathbb{P}(\Phi_X(\pf(A))(\Bar{r}, \Bar{s}) = 0) \leqslant 1/2$.

We have 
$$\Phi_X(\pf(A))(\Bar{r}, \Bar{s}) = \sum\limits_{I \subseteq [k]} \pf(A)_{-I}(\Bar{r}, \Bar{s}).$$

It is well known that $\pf(A)^2 = \det(A)$ (for a proof see e.g. \cite{Rote2001-ra}).
Since in a field of characteristic $2$, for any $x \in F$, we have $x = -x$, and also $\sqrt x$ exists and is unique. Therefore, we can write
$$\Phi_X(\pf(A))(\Bar{r}, \Bar{s}) = \sum\limits_{I \subseteq [k]} \sqrt{\det(A)_{-I}(\Bar{r}, \Bar{s})}.
$$

For every $I \subseteq [k]$, let us define an auxiliary $n \times n$ matrix $B_I$ over $\mathbb{F}$ such that
$$
B_I[i,j] = A[i,j]_{-I}(\Bar{r}, \Bar{s})
$$
for every $i, j \in [N]$.
It is easy to see that $\det(A)_{-I}(\Bar{r}, \Bar{s}) = \det(B_I)$. Hence,
\begin{equation}\label{eq:pfdet}
\Phi_X(\pf(A))(\Bar{r}, \Bar{s}) = \sum\limits_{I \subseteq [k]} \sqrt{\det(B_I)}.
\end{equation}

In order to compute the expression \eqref{eq:pfdet}, we need to know the matrices $B_I$ for all $I \subseteq [k]$.
However, there are $2^k$ such matrices and each entry in $A_{-I}$ is a polynomial built of up to $2^k$ monomials, making the naive evaluation far too slow.
Fortunately, for a fixed $i, j \in [N]$, we can compute all the values $B_I[i,j]$ simultaneously and efficiently, using dynamic programming, as follows.

Let us fix $i, j \in [N]$ and define function $f : 2^{[k]} \to \mathbb{F}$ such that, for every $J \subseteq [k]$,
\[f(J) = P_{i, j, J}(\Bar{r}, \Bar{s}).\]
It is straightforward to compute all the values of $f$ in $O(2^k k)$ arithmetic operations.
Note that $\mathbb{F}$ is isomorphic to $F_2[x]$ modulo an irreducible polynomial of degree~$q$. Hence, any arithmetic operation takes $O(q \operatorname{polylog}(q)) = O(\operatorname{polylog}(n))$.
This yields $O(2^k k \operatorname{polylog}(n))$ time for computing all the values of $f$.

Now, let us define function $g : 2^{[k]} \to \mathbb{F}$ such that, for every $I \subseteq [k]$, $g(I) = \sum_{J \subseteq I} f(J)$.
It is now easy to see that $B_I[i,j] = g([k]\setminus I)$.
Further, $g$ is the so called \emph{zeta transform} of $f$, and can be computed in $O(2^k k)$ arithmetic operations~\cite{yates1937design} using dynamic programming. This leads to $O(2^k k \operatorname{polylog}(n))$ time in our case. 

We can therefore compute the functions $f$ and $g$ for all $i,j \in [N]$, and consequently find the matrices $B_I$ for all $I \subseteq [k]$, in $O(2^k k n^2\operatorname{polylog}(n))$ total time.

It is well known that the determinant of an $n \times n$ matrix with entries from a field can be computed in $O(n^{\omega + o(1)})$ arithmetic operations (see, e.g., \cite{DBLP:books/aw/AhoHU74}, Chapter 6), where $\omega < 2.37286$ is the \emph{matrix multiplication exponent} \cite{DBLP:conf/soda/AlmanW21}. 
This means that we can evaluate the expression \eqref{eq:pfdet} using additional $O(2^k n^{\omega + o(1)} )$ time.
\end{proof}

\end{document}